\def\beq{\begin{equation}}
\def\ee{\end{equation}}

\documentclass[11pt, reqno]{amsart}
\usepackage{amsmath}
\usepackage{amscd}
\usepackage{amsthm}
\usepackage{amssymb} \usepackage{latexsym}
\usepackage{eufrak}
\usepackage{euscript}
\usepackage{epsfig}
\usepackage{graphics}
\usepackage{array}
\usepackage{enumerate}
\usepackage{dsfont}
\usepackage{color}
\usepackage{wasysym}
\usepackage{braket}
\usepackage{pmat}
\usepackage{pdfsync}

\newcommand{\bel}[1]{\begin{equation}\label{#1}}

\newcommand{\ba}{\begin{eqnarray}}
\newcommand{\ea}{\end{eqnarray}}

\newcommand{\qe}{\end{equation}}
\newcommand{\R}{{\mathbb R}}
\newcommand{\N}{{\mathbb N}}

\newcommand{\C}{{\mathbb C}}
\newcommand{\BH}{\mathcal{H}}

\newcommand{\al}{{\alpha}}

\theoremstyle{theorem}
\newtheorem{thm}{Theorem}[section]

\theoremstyle{example}
\newtheorem{example}{Example}[section]
\theoremstyle{corollary}

\theoremstyle{lemma}

\theoremstyle{definition}

\theoremstyle{Assumption}
\newtheorem{assumption}{Assumption}[section]
\theoremstyle{proof}

\theoremstyle{remark}
\newtheorem{rem}{Remark}[section]

\begin{document}

\title[Schmidt-correlated states, weak Schmidt decomposition and generalized Bell basis]{Schmidt-correlated states,
weak Schmidt decomposition and generalized Bell bases related to Hadamard matrices}

\author{Bobo Hua}
\address{Max Planck Institute for Mathematics in the Sciences, 04103 Leipzig,
 Germany, and School of Mathematical Sciences, LMNS, Fudan University, Shanghai 200433, China}
\email{bobohua@fudan.edu.cn}
\author{Shaoming Fei}
\address{Max Planck Institute for Mathematics in the Sciences, 04103 Leipzig,
Germany, and School of Mathematical Sciences, Capital Normal
University, Beijing 100048, China} \email{feishm@cnu.edu.cn}
\author{J\"urgen Jost}
\address{Max Planck Institute for Mathematics in the Sciences, 04103 Leipzig,
 Germany,
Department of Mathematics and Computer Science, University of
Leipzig, 04109 Leipzig, Germany, and the Santa Fe Institute for the
Sciences of Complexity, Santa Fe, NM 87501, USA} \email{jost@mis.mpg.de}
\author{Xianqing Li-Jost}
\address{Max Planck Institute for Mathematics in the Sciences,  04103 Leipzig, Germany, and
School of Mathematics and Statistics, Hainan Normal University,
571158 Haikou, China} \email{xli-jost@mis.mpg.de}

\maketitle

\begin{abstract}
We study the mathematical structures and relations among some quantities in the theory of quantum entanglement,
such as separability, weak Schmidt decompositions, Hadamard matrices etc..
We provide an operational method to identify the Schmidt-correlated
states by using weak Schmidt decomposition. We show that a mixed state is Schmidt-correlated
if and only if its spectral decomposition consists of a set of pure eigenstates
which can be simultaneously diagonalized in weak Schmidt decomposition,
i.e. allowing for complex-valued diagonal entries. For such states, the separability
is reduced to the orthogonality conditions of the vectors consisting of diagonal
entries associated to the eigenstates, which is surprisingly related to the so-called complex Hadamard matrices.
Using the Hadamard matrices, we provide a variety of generalized maximal entangled Bell bases.
\end{abstract}

\medskip

\section{Introduction}
As one of the most striking features of quantum systems, quantum
entanglement \cite{HorodeckiHHH} plays crucial roles in quantum
information processing \cite{nielsen} such as quantum computation,
quantum teleportation, dense coding, quantum cryptographic schemes,
quantum radar, entanglement swapping and remote states preparation.
Nevertheless, many significant open problems in characterizing
the entanglement of quantum systems still remain open.


Let $\mathcal{H}=\mathcal{H}_A\otimes\mathcal{H}_B\simeq
\C^n\otimes \C^n$ be a bipartite composite system.
The mathematical problem consists in deriving separability
criteria of mixed quantum states, and more generally, in quantifying their
degree of entanglement. Such basic  problems in the theory of
quantum entanglement turn out to be surprisingly difficult,
see e.g. \cite{BZ} for a monographic treatment. Reasons for the
difficulty are that the representations of a mixed state $\rho$ as a (statistical) ensemble
of pure states are not unique, and that the pure states in
a representation in general cannot be simultaneously
diagonalized in terms of suitable bases of $\mathcal{H}_A$ and $\mathcal{H}_B$ (Schmidt decomposition).

One of the main idea of the paper is that one
can weaken the requirement of simultaneous Schmidt diagonalization to a more
general complex version, that is, the generalized Schmidt coefficients are allowed to be
complex-valued, see Section 2. We call it the \emph{weak Schmidt
  decomposition}. It was first introduced by \cite{HiroshimaHayashi} to study quantum states.
In our applications, this concept will invoke the Hadamard matrices, a class of matrices that already have received
considerable mathematical attention, though some basic
problems still remain unresolved, see \cite{TadejZyczkowski} for a survey.

We first recall some basic concepts in the theory of quantum entanglement.
The entanglement of formation
\cite{BDSW,Horo-Bruss-Plenioreviews1,Horo-Bruss-Plenioreviews2,Horo-Bruss-Plenioreviews3}
and concurrence \cite{Wootters98,s7,s8,s9} are among the  important
measures to quantify the entanglement. However, due to the
extremizations involved in the computation, only a few analytic
formulae have been obtained for states like two-qubit ones
\cite{Wootters98,Hill-Wootters97}, isotropic states \cite{Terhal}
and Werner states \cite{Vollbrecht}. Instead of analytic formulas,
some progress has been made toward the lower and upper bounds
\cite{Chen1,Chen2,MintertKB,deVicente,LiX,LiMing,CH,zhxn,63,135}.


A mixed state $\rho$ is called \emph{Schmidt-correlated}, or called \emph{maximally correlated}
\cite{Rains99,Rains01,HiroshimaHayashi}, if there exists an orthonormal basis, $\{\Ket{e_jf_l}\}_{j,l=1}^n,$ of $\BH$ such that
\begin{equation}\label{e:def Schmidt correlated}
\rho=\sum_{j,l=1}^n C_{jl}\Ket{e_jf_j}\Bra{e_lf_l}.
\end{equation}
It is called maximally correlated since for
any classical measurement on $\BH_A$ or $\BH_B$, Alice and Bob will always
obtain the same result. It turns out that this class of states exhibits many excellent properties \cite{Rains99,VedralPleniio98,VirmaniSacchiPlenioMarkham01,ChenYang,HiroshimaHayashi,KhasinKosloffSteinitz07,ZhaoFeiWang}.
However, given a general state $\rho$ written in the computational basis,
any operational method to decide whether it is Schmidt-correlated is still missing in the literature.
In this paper, we show that to decide a Schmidt-correlated state it suffices to check whether its
spectral decomposition consists of pure eigenstates which can be simultaneously
diagonalized in weak Schmidt decomposition, see Theorem \ref{thm:schmidt correlated}.
Although the spectral decomposition may not be unique in the case that it possesses
eigenvalues of high multiplicity and it is very possible that the property of simultaneously
diagonalization in weak Schmidt decomposition strongly depends on the choice of the ensembles
(or eigenstates), our theorem indicates that it is sufficient to verify only one of the ensembles.
In fact, one can derive that all ensembles of a mixed state can be
simultaneously diagonalized in weak Schmidt decomposition if and only if one of them
can be, for which we present a direct proof by Schr\"odinger's mixture theorem (Theorem 8.2 in \cite{BZ}), see \cite{ZhaoFeiWang} for an alternative proof. For the criteria of this
simultaneous diagonalization it boils down to the standard matrix theory, see \cite{Wie,Gib,HiroshimaHayashi}.
In this way, we provide an operational method to solve the problem.

Generally, it is of significance to find a complete basis of maximal entanglement, e.g.
the Bell states for qubits. In this paper, we use complex Hadamard matrices to introduce a wide class
of bases consisting of generalized Bell states of maximal entanglement, which contains the
well-known Weyl operator basis \cite{BBCJPW,Na,BHN06,BHN07,BK,HiroshimaHayashi,LiuBaiGeJing13} as a special case.

The paper is organized as follows. In section~\ref{s:weak schmidt},
we introduce the main concept of the paper, simultaneous
diagonalizations in weak singular value decomposition (weak Schmidt decomposition).
In Section~\ref{s:schmidt correlated}, we prove our main result, Theorem \ref{thm:schmidt correlated},
and use it to identify all Schmidt-correlated states. Section~\ref{s: entanglement of
convex combinations} is devoted to the separability criteria of Schmidt-correlated states.
In section~\ref{s:generalized bell basis with
complex Hadamard matrices}, we explore the deep connections among
the separability criterion, Hadamard matrices and generalized Bell
bases. Conclusion and remarks are given in the last section.

\section{Weak Schmidt decomposition}\label{s:weak schmidt}
Let $\mathcal{H}=\mathcal{H}_A\otimes\mathcal{H}_B\simeq
\C^n\otimes \C^n$ be a bipartite composite system and $\{\Ket{jl}\}_{j,l=1}^n$ the
computational basis of $\BH.$ In this basis, any pure state can be written as
$\Ket{\psi}=\sum_{j,l=1}^na_{jl}\Ket{jl}$ which associates with a matrix $A=(a_{jl})_{n\times n}.$
The Schmidt decomposition asserts that there exists
an orthonormal basis of $\BH,$ $\{\Ket{e_jf_l}\}_{j,l=1}^n,$ such that
$\Ket{\psi}=\sum_{j=1}^n{\sqrt{\lambda_j}}\Ket{e_j}\Ket{f_j}$ where
$\lambda_j\geq 0$ and $\sum_{j=1}^n\lambda_j=1.$ This follows from the singular value decomposition, SVD in short, of the matrix
$A,$ i.e. there exist $n\times n$ unitary matrices $U$ and $V$ such that
\[UAV^{t}=\begin{pmatrix}
\sqrt{\lambda_1}&&&\\
&\sqrt{\lambda_2}&&\\
&&\ddots&\\
&&&\sqrt{\lambda_n}
\end{pmatrix}_{n\times n},\] where $t$ denotes the transpose.

A pure state $\Ket{\psi}$ is called \emph{separable} if it is a
product state, i.e.
$\Ket{\psi}=\Ket{\psi_1}\otimes\Ket{\psi_2},\ \Ket{\psi_1}\in\BH_A,\ \Ket{\psi_2}\in\BH_B.$ A \emph{mixed state} $\rho$ is a statistical ensemble of pure
  states, denoted by $\{p_k,\Ket{\psi_k}\}_{i=1}^K$ with $p_k> 0$
  and $\sum_{k=1}^Kp_k=1$, that is, \begin{equation}\label{rho1}\rho:=\sum_{k=1}^Kp_k\Ket{\psi_k}\Bra{\psi_k}.\end{equation}
 A mixed state
$\rho$ is \emph{separable} if it can be expressed as a convex combination of
separable pure states, i.e., there exists an ensemble of separable pure states, otherwise it is called entangled \cite{Werner}. In general, it
is difficult to decide whether a given mixed state is separable or
not, because the ensembles of a mixed state are generically non-unique.
There are neither operational sufficient and necessary criteria for judging the
separability in general, nor analytical formulae for entanglement of
formation or concurrence for arbitrary mixed states.  However, if one could
carry out the Schmidt decomposition simultaneously for all the pure
states in an ensemble of a mixed state, then the calculation would
be much easier. Therefore the question is under which conditions one can
diagonalize a set of pure states in SVD simultaneously. This is
answered by Wiegmann's theorem \cite{Wie}: A
set of matrices $\{A_k\}_{k=1}^K$ can be simultaneously diagonalized
in SVD iff for any $1\leq i,j\leq K$
$$A_i A^{\dag}_j = A_j A^{\dag}_i\ \mathrm{and}\ A_i^{\dag} A_j = A_j^{\dag} A_i.$$

However, this demand is too strong for our purposes, as is already evident from the following
example,
 \[\Ket{\psi_1}\sim A_1=\begin{pmatrix}
 1&&\\
 &\omega&\\
 &&\omega^2
 \end{pmatrix},\ \ \  \Ket{\psi_2}\sim A_2=\begin{pmatrix}
 1&&\\
 &\omega^2&\\
 &&\omega
 \end{pmatrix}, \ \]
 where $\omega=\frac{1}{2}+i\frac{\sqrt{3}}{2}.$ Direct calculation
 shows that $A_1$ and $A_2$ cannot be transformed to SVD simultaneously, although they are already in complex diagonal form.

Our idea is to investigate when the pure states of an ensemble of
a mixed state can be simultaneously diagonalized in complex-valued
form. Namely, we consider the more general case (than SVD) in which we allow
complex-valued entries for the diagonal matrices. We say that $\{\Ket{\psi_k}\}_{k=1}^K$ can be
\emph{simultaneously diagonalized in weak SVD} if there exist
$n\times n$ unitary matrices $U$ and $V$ such that $UA_kV^{t}$
($1\leq k\leq K$) are complex-valued diagonal matrices, where $A_k$ are the matrix
representations of $\Ket{\psi_k}.$ This kind of diagonalization can be regarded as a
``\emph{weak Schmidt decomposition}."
A mixed state $\rho$ is called \emph{simultaneously diagonalizable in weak SVD} if there exists an ensemble $\{p_i,\Ket{\psi_i}\}$ of $\rho$ such that $\{\Ket{\psi_i}\}$ can be simultaneously diagonalized in weak SVD. The previous example shows that this simultaneous
diagonalization is really weaker than classical
simultaneous diagonalization in SVD.

By matrix theory, see Wiegmann \cite{Wie} and Gibson \cite{Gib}, the
set of matrices $\{A_k\}_{k=1}^K$ can be simultaneously diagonalized in weak SVD
if and only if
\begin{equation}\label{e:simu 1}
A_jA_k^{\dagger}A_l=A_lA_k^{\dagger}A_j,~~~~\forall
1\leq j,k,l\leq K,
\end{equation}
or if and only if $A_k^{\dagger}A_l$ is normal and
\begin{equation}\label{e:simu 2}
A_jA_k^{\dagger}A_kA_l^{\dagger}=A_kA_l^{\dagger}A_jA_k^{\dagger},~~~~\forall
1\leq j,k,l\leq K.
\end{equation}
This simultaneous diagonalization
was first introduced to study quantum states by \cite{HiroshimaHayashi}.

The following example (whose general nature will become apparent
  in Section \ref{s:generalized bell basis with complex Hadamard
matrices})  shows the advantage of this generalized
diagonalization. Consider a $3\times 3$ mixed state, \begin{equation*}\label{p1}
\rho:=\sum_{k=1}^3 p_k\Ket{\psi_k}\Bra{\psi_k}, \end{equation*} where
$$
\Ket{\psi_1} =\frac{1}{3} (1,1,1|1,\omega,\omega^2| 1,\omega^2,
\omega)^t ,
$$
$$
\Ket{\psi_2}=\frac{1}{3} (1,\omega,\omega^2| 1,\omega^2, \omega|
1,1,1)^t
$$
and
$$ \Ket{\psi_3}= \frac{1}{3}(1,\omega^2,\omega| 1,1,1|
1,\omega, \omega^2)^t,
$$
$p_k> 0$, $1\leq k\leq 3$ and $\sum_{k=1}^3 p_k=1$. While $\{\Ket{\psi_i}\}_{i=1}^3$ have the same singular values, they cannot be
simultaneously diagonalized in SVD in the classical sense. However, they can be
simultaneously diagonalized in weak SVD through
$$
U= \frac{1}{\sqrt{3}}
\begin{pmatrix}
1&1&1\\
1&\omega&\omega^2\\
1&\omega^2&\omega\
\end{pmatrix}^{\dag},\ \ \ V=I,
$$
where $I$ stands for the identity matrix. Corresponding to this
simultaneous diagonalization, one has \begin{equation*}\label{p2}
\rho:=\sum_{k=1}^3 p_k\Ket{\phi_k}\Bra{\phi_k} , \end{equation*} where
$$
\Ket{\phi_1} =\frac{1}{\sqrt{3}} (1,0,0| 0,1,0| 0,0,1)^t,
$$
$$
\Ket{\phi_2}=\frac{1}{\sqrt{3}} ( 1,0,0| 0,\omega, 0|
0,0,\omega^2)^t,
$$
$$
\Ket{\phi_3}= \frac{1}{\sqrt{3}}(1,0,0| 0,\omega^2, 0|
0,0,\omega)^t.
$$
Obviously this ensemble has a much clearer internal structure than the previous one.
We shall discuss such states in detail in next sections.

\section{Schmidt-correlated states and simultaneous diagonalization in weak SVD}\label{s:schmidt correlated}
Schmidt-correlated states have proven to be quite useful, see \cite{Rains99,VedralPleniio98,VirmaniSacchiPlenioMarkham01,ChenYang,HiroshimaHayashi,KhasinKosloffSteinitz07,ZhaoFeiWang}.
However, given a general state $\rho$ written in the computational basis $\{\Ket{jl}\}_{1\leq j,l\leq n},$
$$\rho=\sum_{i,j,k,l}\rho_{ij,kl}\Ket{ij}\Bra{kl},$$
it is a hard problem to decide whether it is Schmidt-correlated. In this section, we provide an operational method to solve this problem.

\begin{thm}\label{thm:schmidt correlated}For a mixed state $\rho,$ the following are equivalent:
\begin{enumerate}[(a)]
\item $\rho$ is Schmidt-correlated.
\item $\rho$ is simultaneously diagonalizable in weak SVD, i.e. there exists an ensemble of $\rho,$ $\{p_k,\Ket{\psi_k}\}_{k=1}^K$, such that $\{\Ket{\psi_k}\}$ is simultaneously diagonalized in weak SVD.
\item For all ensembles of $\rho,$ $\{q_s,\Ket{\phi_s}\}_{s=1}^S$, $\{\Ket{\phi_s}\}$ are simultaneously diagonalized in weak SVD.
\end{enumerate}
\end{thm}
\begin{proof}
$(a)\Longrightarrow (b)$: Let $\rho$ be of the form
\eqref{e:def Schmidt correlated}. One can show that the matrix
$C=(C_{jl})_{n\times n}$ is positive semidefinite and
has trace $1$.  Hence the spectral theorem implies that
$$C_{jl}=\sum_{k=1}^n\lambda_kv_k^j(v_k^l)^*,\ \ \ 1\leq j,l\leq n,$$
where
$\lambda_k$ are the eigenvalues of $C$ satisfying $\lambda_k\geq 0$
and $\sum_{k=1}^n\lambda_k=1,$ and
$(v_k^1,v_k^2,\cdots,v_k^n)^t$ are the normalized eigenvectors
pertaining to $\lambda_k.$ This yields
$$\rho=\sum_{k=1}^n\lambda_k\left(\sum_{j=1}^nv_k^j\Ket{e_jf_j}\right)\left(\sum_{l=1}^n(v_k^l)^*\Bra{e_lf_l}\right).$$
By setting $p_k=\lambda_k$ and
$\Ket{\psi_k}=\sum_{j=1}^nv_k^j\Ket{e_jf_j},$ we prove $(b).$

$(b)\Longrightarrow(a)$: This follows from direct computation.

$(c)\Longrightarrow (b)$: This is trivial.

$(a)\Longrightarrow (c)$:
This follows directly from Schr\"odinger's mixture theorem, Theorem 8.2 in \cite{BZ}. (One can also show this by an argument in \cite{ZhaoFeiWang}.)  By $``(a)\Longrightarrow (b)"$ above, the ensemble of eigenstates of $\rho=\sum_{k=1}^n\lambda_k\Ket{\psi_k}\Bra{\psi_k}$ is simultaneously diagonalizable. Schr\"odinger's mixture theorem implies that for any ensemble of $\rho,$ $\{q_s,\Ket{\phi_s}\}_{s=1}^S,$ there exists an $S\times S$ unitary matrix $U$ such that
$$\Ket{\phi_s}=\frac{1}{\sqrt{q_s}}\sum_{k=1}^nU_{sk}\sqrt{\lambda_k}\Ket{\psi_k}.$$ Since $\{\Ket{\psi_k}\}_{k=1}^n$ is simultaneously diagonalizable, so is $\{\Ket{\phi_s}\}_{s=1}^S.$
\end{proof}

This theorem suggests an operational method to check whether a mixed state $\rho$ is Schmidt-correlated.
First at all, we write $\rho$ in the spectral decomposition, i.e.
$\rho=\sum_{k=1}^K\lambda_k\Ket{\psi_k}\Bra{\psi_k}$ where $\{\lambda_k\}_{k=1}^K$,
are eigenvalues of $\rho$ and $\{\Ket{\psi_k}\}_{k=1}^K$ are the corresponding eigenstates.
Although the spectral decomposition may not be unique, e.g. the eigenvalues has high multiplicity,
by our result it suffices to check whether a particular ensemble $\{\Ket{\psi_k}\}$ is
simultaneously diagonalized in weak SVD. As we know, this reduces to the criteria of the
simultaneous diagonalization in weak SVD by Wiegmann \cite{Wie} and Gibson \cite{Gib},
see \eqref{e:simu 1} or \eqref{e:simu 2}. Furthermore, one can use the process of simultaneous
diagonalization to calculate the basis $\{\Ket{e_jf_l}\}$ which diagonalizes $\{\Ket{\psi_k}\}.$


\section{Separability of convex combinations of simultaneously diagonalizable states}\label{s: entanglement of convex combinations}
While the problem of separability of a general mixed state is hard, for Schmidt-correlated states,
i.e. simultaneously diagonalizable in weak SVD states, the criteria of separability is quite
simple once we know the basis $\{\Ket{e_jf_l}\}$ which diagonalizes the state.
Let $\rho$ be a mixed state of ensemble $\{p_k,\Ket{\psi_k}\}_{k=1}^K,$
\begin{equation}\label{rho}
\rho:=\sum_{k=1}^Kp_k\Ket{\psi_k}\Bra{\psi_k},
\end{equation}
where
$\{\Ket{\psi_k}\}_{k=1}^K$ can be simultaneously diagonalized in a new orthonormal basis
$\{\Ket{e_jf_l}\}_{1\leq j,l\leq n}$ such that
$$
\Ket{\psi_k}:=\sum_{j=1}^n\al_{j,k}\Ket{e_jf_j}
\sim A_k=\begin{pmatrix}
\al_{1,k}&&&\\
&\al_{2,k}&&\\
&&\ddots&\\
&&&\al_{n,k}
\end{pmatrix}_{n\times n}
$$
with diagonal entries $\al_{j,k}\in\C$ and
$\sum_{j=1}^n|\al_{j,k}|^2=1$. The
following theorem gives a necessary and sufficient condition for the
separability of $\rho.$

\begin{thm}\label{main}
Let $\rho$ be a Schmidt-correlated state written in the orthonormal basis $\{\Ket{e_jf_l}\}$
$$\rho=\sum_{j,l=1}^nC_{jl}\Ket{e_jf_j}\Bra{e_lf_l}.$$ Then the following are equivalent:
\begin{enumerate}[(a)]
\item $\rho$ is separable.
\item $\rho$ is PPT (positive partial transposition).
\item $C_{jl}=0, \ \ \ \forall j\neq l.$
\item For any ensemble $\{p_k,\Ket{\psi_k}\}_{k=1}^K$ of $\rho,$ denoted by $\Ket{\psi_k}:=\sum_{j=1}^n\al_{j,k}\Ket{e_jf_j},$ we have \begin{equation}\label{thetacon}\sum_{k=1}^Kp_k\al_{j,k}\al_{l,k}^*=0\end{equation} for all $j\neq l,$
$1\leq j,l\leq n$, where $*$ denotes complex conjugate.
\item For any ensemble of $\rho$, \eqref{thetacon} holds.
\end{enumerate}
\end{thm}

\begin{proof}
Direct computation shows that $(c),(d),(e)$ are equivalent.

$(a)\Longrightarrow(b)$: This follows from a theorem of
Peres \cite{Peres}.

$(b)\Longrightarrow(c)$: We use a contradiction
argument. Suppose there exist $1\leq j_0< l_0\leq n$ such that
$C_{j_0l_0}\neq0,$ then we will show
that $\rho$ is NPPT (non-positive partial transposition). In fact, taking the partial transpose of the
second subsystem of $\rho$ yields
\begin{equation*}
\rho^{T_B}=\sum_{1\leq j,l\leq n}C_{jl}\Ket{e_jf_l}\Bra{e_lf_j}.
\end{equation*}
Since one of the principal minors of order two of $\rho^{T_B}$ reads
as \[\left|\begin{matrix}(\rho^{T_B})_{j_0l_0,j_0l_0}&(\rho^{T_B})_{j_0l_0,l_0j_0}\\
(\rho^{T_B})_{l_0j_0,j_0l_0}&(\rho^{T_B})_{l_0j_0,l_0j_0}\end{matrix}\right|=-\left|C_{j_0l_0}\right|^2<0,\]
we have $\rho^{T_B}\not\geq 0.$ Hence, $\rho$ is NPPT.

$(c)\Longrightarrow(a)$: Obviously, $\rho=\sum_{j}C_{jj}\Ket{e_jf_j}\Bra{e_jf_j}$ is
separable.
\end{proof}

\section{Generalized Bell bases and complex Hadamard matrices}\label{s:generalized bell basis with complex Hadamard
matrices}

In this section, we introduce  generalized Bell
bases and explore their connections with Hadamard matrices and the
separability criterion. This is motivated by the fact that the
computational basis $\Ket{jl}$, while usually used when
investigating entanglement, is not always the most suitable one for
our purposes. In particular, mixed states that have ensembles of pure
states which can be simultaneously diagonalized should rather be investigated in that
diagonalized form as this involves the least number of parameters. As the separability concerns the question to what extent a state behaves like a product
state, the non-separability is however
quantified by entanglement. Hence one should look at the length of the projection of a
given state onto the maximally entangled states. Therefore we introduce new kind of
bases consisting of maximal entangled states, called generalized Bell bases
(see \eqref{e:Bell basis} below), which contains the specific Bell-states for the case of qubit pairs.


First we want to find nontrivial  solutions of the system of
equations \eqref{thetacon}. We restrict to the following class of
states.
\begin{assumption}\label{a: assumptions for rho} The Schmidt-correlated state $\rho$
written in \eqref{rho} satisfies
\begin{enumerate}[(a)]\item $K=n,$ \item $|\al_{j,k}|=a_j\neq 0$ for any $1\leq j,k\leq n,$
\item $p_k=1/n$ for all $1\leq k\leq n.$
\end{enumerate}
\end{assumption}
The normalization conditions require that $\sum_{j=1}^n a_j^2=1$. In polar
coordinates, one can write $\al_{j,k}=a_je^{i\theta_{j,k}}.$ In
fact, we will figure out soon that only the phases $(e^{\theta_{j,k}})$
matter for the separability of $\rho$ in this case. In characterizing all the solutions of
\eqref{thetacon}, the key observation is that the orthogonality
conditions of Theorem \ref{main} translates into the conditions for a
complex Hadamard matrix.
Here, an $n\times n$ complex matrix $H$ is called a \emph{complex Hadamard
matrix}, see \cite{TadejZyczkowski,Butson1,Butson2}, if
$|H_{j,k}|=1$ for all $1\leq j,k\leq n$ and $HH^\dagger=n I.$
Equivalently, $\frac{1}{\sqrt{n}}H$ is a unitary matrix.

\begin{thm}
Let $\rho$ be a Schmidt-correlated state which is an ensemble of
$\{p_k,\Ket{\psi_k}\}_{k=1}^K$ with $\Ket{\psi_k}=\sum_{j=1}^K\al_{j,k}\Ket{e_jf_j}$
for an orthonormal basis $\Ket{e_jf_l}$ satisfying Assumption \ref{a: assumptions for
rho}. Then $\rho$ is separable if and only if
$(e^{i\theta_{j,k}})_{n\times n}$ is a complex Hadamard matrix, where $e^{i\theta_{j,k}}$ are the phase factor of $\al_{j,k}$.
\end{thm}
\begin{proof}
It is obvious that $|e^{i\theta_{j,k}}|=1$ for all $1\leq j,k\leq
n.$ And the system of equations \eqref{thetacon} is equivalent to
the orthogonality of the rows of $(e^{i\theta_{j,k}})_{n\times n}.$
The result follows from Theorem \ref{main}.
\end{proof}

\begin{rem}
The moduli of the diagonal entries, $a_j,$ in Assumption \ref{a:
assumptions for rho} can be chosen arbitrarily (as long as
the normalization condition $\sum_{j=1}^n a_j^2=1$ holds). This theorem
indicates that when the phase factors $(e^{i\theta_{j,k}})$
constitute a complex Hadamard matrix, no matter what $a_j$ are, the state $\rho$ becomes separable.
\end{rem}

By the theory of complex Hadamard matrices, there always exists a
solution to \eqref{thetacon} for any $n\in \N$ under Assumption
\ref{a: assumptions for rho}. For instance, for any $n$ the Fourier
matrix $F_n:=(H_{j,k})_{n\times
n}:=(e^{i(j-1)(k-1)\frac{2\pi}{n}})_{n\times n}$ is an Hadamard
matrix. Two Hadamard matrices, $H_1$ and $H_2$, are called
equivalent if there exist diagonal unitary matrices $D_1$ and $D_2$
and permutation matrices $P_1$ and $P_2$ such that
$$H_1=D_1P_1H_2P_2D_2.$$

For the classification up to this equivalence of Hadamard matrices for
$n\leq 5$,  see e.g. Tadej and Zyczkowski
\cite{TadejZyczkowski}. For instance, the Fourier matrix is the only
Hadamard matrix for $n=3$ and $n=5.$ For $n=4,$ there is a
continuous non-equivalent family of Hadamard matrices. For $n\geq
6,$ things become more complicated. In
particular, a complete classification of Hadamard matrices of order $6$ is still unknown.

For $n=3$, $\rho=1/n\sum_{k=1}^n\Ket{\psi_k}\Bra{\psi_k}$ is then
separable if and only if, up to the equivalence of Hadamard matrices,
$$(e^{i\theta_{j,k}})_{3\times 3}=\begin{pmatrix}
1&1&1\\
1&\omega&\omega^2\\
1&\omega^2&\omega
\end{pmatrix},
$$ where $\omega=\frac{1}{2}+i\frac{\sqrt{3}}{2},$ i.e.,
$$\Ket{\psi_1}\sim\begin{pmatrix}
a_1&&\\
&a_2&\\
&&a_3\\
\end{pmatrix}
,\Ket{\psi_2}\sim\begin{pmatrix}
a_1&&\\
&a_2\omega&\\
&&a_3\omega^2
\end{pmatrix},
\Ket{\psi_3}\sim\begin{pmatrix}
a_1&&\\
&a_2\omega^2&\\
&&a_3\omega\\
\end{pmatrix}.
$$

In higher dimensions, there are more freedom for the existence of Hadamard
matrices and the corresponding constructions of separable states.
Our result connnects the separability problem to the study of Hadamard
matrices.

We now construct generalized Bell bases for
$\mathcal{H}_A^n\times \mathcal{H}_B^n$ by Hadamard matrices. Let
\begin{eqnarray}\label{e:Bell basis}
\Ket{\psi_l^1}&=&\frac{1}{\sqrt{n}}\sum_{j=1}^ne^{i\phi_{j,l}^1}\Ket{j,j},\nonumber\\
\Ket{\psi_l^2}&=&\frac{1}{\sqrt{n}}\sum_{j=1}^ne^{i\phi_{j,l}^2}\Ket{j,j+1},\\
&&\cdots\cdots\nonumber\\
\Ket{\psi_l^n}&=&\frac{1}{\sqrt{n}}\sum_{j=1}^ne^{i\phi_{j,l}^n}\Ket{j,j+n-1},\nonumber
\end{eqnarray}
where we count $j+n-1 \text{ mod }n$, $1\leq l\leq n$, and $(e^{i\phi_{j,l}^s})_{j,l}$ is
a Hadamard matrix for any fixed $1\leq s\leq n$. One finds that $\left\{\Ket{\psi_{l}^k}\right\}_{1\leq
l,k\leq n}$ constitute an orthonormal basis of
$\mathcal{H}_A^n\times \mathcal{H}_B^n$. By using Hadamard matrices, we know that
$\ket{\psi_l^k}$, $1\leq l,k\leq n$, is maximally entangled. Therefore, any mixed state $\rho$
can be written as
$$\rho=\sum_{1\leq l,k,m,j\leq n}\rho_{lk,mj}\Ket{\psi_{l}^k}\Bra{\psi_{m}^j}.$$

When we choose $(e^{\phi_{j,l}^s})_{n\times n}=F_n$
in \eqref{e:Bell basis} for  $1\leq s\leq n$, where $F_n$ is the
Fourier matrix of order $n$, we  recover the well-known Weyl
operator basis. The Weyl operators have been introduced in the
context of quantum teleportation \cite{BBCJPW} and investigated
thoroughly in the literature (see e.g. \cite{Na,BHN06,BHN07,BK}).
Note that our generalized Bell bases possess more freedom, as we may choose any complex Hadamard matrices.
In higher dimensions, there exist plenty of complex Hadamard matrices, offering a
potential for new applications. Let us consider an example:

\begin{example}
For $n=4$, besides the Weyl basis there exist many other Bell bases. Let
$$(e^{\phi_{j,l}^s})_{4\times 4}=\begin{pmatrix}
1&1&1&1\\
1&ie^{ia_s}&-1&-ie^{ia_s}\\
1&-1&1&-1\\
1&-ie^{ia_s}&-1&ie^{ia_s}\\
\end{pmatrix},
$$
where $1\leq s\leq 4$ and $a_s\in \R.$ These are Hadamard matrices,
see $(65)$ in Section $5.4$ \cite{TadejZyczkowski}. Then for any $a_s\in \R,$ $1\leq s\leq 4$,  we have
the following Bell basis:
{\small $$\begin{pmatrix}
1&&&\\
&1&&\\
&&1&\\
&&&1\\
\end{pmatrix}
,\begin{pmatrix}
1&&&\\
&ie^{ia_1}&&\\
&&-1&\\
&&&-ie^{ia_1}\\
\end{pmatrix},
\begin{pmatrix}
1&&&\\
&-1&&\\
&&1&\\
&&&-1\\
\end{pmatrix},
\begin{pmatrix}
1&&&\\
&-ie^{ia_1}&&\\
&&-1&\\
&&&-ie^{ia_1}\\
\end{pmatrix};$$

$$
\begin{pmatrix}
&1&&\\
&&1&\\
&&&1\\
1&&&\\
\end{pmatrix}
,\begin{pmatrix}
&1&&\\
&&ie^{ia_2}&\\
&&&-1\\
-ie^{ia_2}&&&\\
\end{pmatrix},
\begin{pmatrix}
&1&&\\
&&-1&\\
&&&1\\
-1&&&\\
\end{pmatrix},
\begin{pmatrix}
&1&&\\
&&-ie^{ia_2}&\\
&&&-1\\
-ie^{ia_2}&&&\\
\end{pmatrix};$$

$$
\begin{pmatrix}
&&1&\\
&&&1\\
1&&&\\
&1&&\\
\end{pmatrix}
,\begin{pmatrix}
&&1&\\
&&&ie^{ia_3}\\
-1&&&\\
&-ie^{ia_3}&&\\
\end{pmatrix},
\begin{pmatrix}
&&1&\\
&&&-1\\
1&&&\\
&-1&&\\
\end{pmatrix},
\begin{pmatrix}
&&1&\\
&&&-ie^{ia_3}\\
-1&&&\\
&-ie^{ia_3}&&\\
\end{pmatrix};
$$

$$\begin{pmatrix}
&&&1\\
1&&&\\
&1&&\\
&&1&\\
\end{pmatrix}
,\begin{pmatrix}
&&&1\\
ie^{ia_4}&&\\
&-1&&\\
&&-ie^{ia_4}&\\
\end{pmatrix},
\begin{pmatrix}
&&&1\\
-1&&&\\
&1&&\\
&&-1&\\
\end{pmatrix},
\begin{pmatrix}
&&&1\\
-ie^{ia_4}&&&\\
&-1&&\\
&&-ie^{ia_4}&\\
\end{pmatrix}.
$$}
\end{example}

In \cite{FeiLi} the case of
$$\rho=\sum_{1\leq j,l\leq
n}\rho_{jl,jl}\Ket{\psi_{l}^j}\Bra{\psi_{l}^j}
$$
for $n=3$ has been investigated. In the
present paper, we have studied in detail the mixed state
\begin{equation}\label{eq1}\rho=\sum_{1\leq l\leq
n}\rho_{jl,jl}\Ket{\psi_{l}^j}\Bra{\psi_{l}^j}\end{equation} for
fixed $j$, $1\leq j\leq n,$ and arbitrary  dimension $n\geq 2.$

\section{Conclusion and remarks}

We have studied the weak Schmidt decomposition of quantum states in
which the generalized Schmidt coefficients are allowed to be
complex-valued. Using this concept, we can identify all so-called Schmidt-correlated states
by verifying the spectral decomposition of the states, hence provide an operational
method to the identification of Schmidt-correlated states. The separability of such states
has been translated into the orthogonality condition of the diagonal entries
of eigenstates of any ensemble. Moreover, we have introduced generalized Bell bases
and presented their connections with Hadamard matrices and separability
criteria. Our construction of generalized Bell
bases includes the well-known Weyl operator basis as a special case.
These mathematical structures and relations among the separability, weak Schmidt decompositions,
Hadamard matrices, generalized Bell bases etc. may help in the further characterization of
quantum entanglement.


%
%

%



\bibliography{Bell}
\bibliographystyle{plain}

\end{document}